\documentclass{article}
\usepackage{tkai}
\usepackage{float} 
\usepackage[T1]{fontenc}
\usepackage[utf8]{inputenc}
\usepackage{times}  
\usepackage{helvet}  
\usepackage{courier}  
\usepackage[hyphens]{url}  
\usepackage{graphicx} 
\usepackage{caption} 
\usepackage{algorithm}
\usepackage{algorithmic}
\usepackage{enumitem}
\usepackage{booktabs}       
\usepackage{amsfonts}       
\usepackage{nicefrac}       
\usepackage{xcolor}         
\usepackage{amsmath}
\usepackage{amssymb}
\usepackage{color}
\usepackage{graphics}
\usepackage{lipsum}

\usepackage{subcaption}
\newtheorem{theorem}{Theorem}[section]
\newtheorem{definition}[theorem]{Definition}
\newtheorem{proposition}[theorem]{Proposition}
\newtheorem{corollary}[theorem]{Corollary}

\newtheorem{lemma}[theorem]{Lemma}

\usepackage{breqn}
\newenvironment{proof}{\paragraph{Proof:}}{\hfill$\square$}
\newcommand{\N}{\mathbb{N}}
\usepackage{newfloat}
\usepackage{listings}
\DeclareCaptionStyle{ruled}{labelfont=normalfont,labelsep=colon,strut=off} 
\floatstyle{ruled}
\newfloat{listing}{tb}{lst}{}
\floatname{listing}{Listing}
\title{On the Tensor Representation and Algebraic Homomorphism of \\the Neural State Turing Machine}
\author{%
  Ankur Mali \\
  University of South Florida\\
  \texttt{ankurarjunmali@usf.edu}
  \And
  Alexander Ororbia \\
  Rochester Institute of Technology \\
  \texttt{ago@cs.rit.edu}
  \AND
  Daniel Kifer \\
  Pennsylvania State University \\
  \texttt{duk17@psu.edu}
  \And
  Lee Giles \\
  Pennsylvania State University \\
  \texttt{clg20@psu.edu}  
}
\usepackage{bibentry}

\begin{document}

\maketitle

\begin{abstract}
Recurrent neural networks (RNNs) and transformers have been shown to be Turing-complete, but this result assumes infinite precision in their hidden representations, positional encodings for transformers, and unbounded computation time in general. In practical applications, however, it is crucial to have real-time models that can recognize Turing complete grammars in a single pass. To address this issue and to better understand the true computational power of artificial neural networks (ANNs), we introduce a new class of recurrent models called the neural state Turing machine (NSTM). The NSTM has bounded weights and finite-precision connections and can simulate any Turing Machine in real-time. In contrast to prior work that assumes unbounded time and precision in weights, to demonstrate equivalence with TMs, we prove that a $13$-neuron bounded tensor RNN, coupled with third-order synapses, can model any TM class in real-time. Furthermore, under the Markov assumption, we provide a new theoretical bound for a non-recurrent  network augmented with memory, showing that a tensor feedforward network with $25$th-order finite precision weights is equivalent to a universal TM. 
\end{abstract}

\section{Introduction}
\label{sec:intro}

There has been a growing interest in improving the memory capability of artificial neural networks when learning algorithms or data structures from examples more efficiently \cite{joulin2015inferring, mali2021recognizing2,mali2020neural, grefenstette2015learning, graves2014neural, graves2016hybrid}. Most artificial neural networks (ANNs) that have been proposed incorporate some type of recurrence \cite{siegelmann94} or attention with positional encodings \cite{attn_turing} and can be seen as Turing complete (TC) by using unbounded memory, weights, and precision. In contrast, recently proposed computational models inspired by Turing machines (TMs), such as the neural Turing machine (NTM) \cite{graves2014neural} and the differentiable neural computer (DNC) \cite{graves2016hybrid}, are not Turing complete in their current form due to the fact that they operate with a fixed-sized memory module and, as a result, can only process space-bounded TMs. An essential requirement for making these models universal, with the capability of processing algorithms of arbitrary length with varied complexity, is to prove an equivalence between such ANNs and Turing machines. 

In contrast to the above, prior work has established a different perspective on deriving the Turing completeness of ANNs. In particular, Siegelman and Sontag \cite{siegelmann95} showed that recurrent neural networks (RNNs) with arbitrary precision and unbounded weights are TC if the model satisfies two particular conditions: 
\textbf{1)} the RNNs can compute or extract dense representations of the data, and 
\textbf{2)} the RNN implements a mechanism to retrieve these stored (dense) representations. 
Notably, this early work demonstrated that, without complex gating and only increasing the model's computational complexity, one can successfully derive the computational power of RNNs. Recently, these types of proofs have been extended to other types of ANNs without recurrence, e.g., transformers and the neural GPU \cite{perez2019turing}. 

The key limitation of the above previously proven TM equivalence for non-recurrent models, such as transformers, is that the underlying proof assumes that the model weights have infinite precision and that the input strings should be encoded as position(al) encodings (in a one-hot encoding format) when provided as input to the ANN. This means that, for an input string alphabet of size one trillion, one would need to create a one-hot vector with a dimensionality of at least one trillion, which is computationally prohibitive. Furthermore, arbitrary precision and weights make it difficult to derive the true computational power of RNNs, thus limiting applicability of these networks to real-time, complex learning scenarios. A similar formal construction has been used to show that bounded RNNs with a growing memory module \cite{chung2021turing} or augmented with a stack-like structure \cite{mali2021recognizing, nndpa1998sun, grefenstette2015learning, joulin2015inferring} are Turing complete and can therefore simulate a TM in O$(T^4)$ with $54$ neurons ($T$ represents the mapping defined by the RNN with parameters $W$ and $b$), which is unrealistic because it assumes unbounded-time and cannot simulate TM in finite or bounded time. However, in this paper, we will formally show that we only require \emph{N} finite precision state neurons to simulate a TM in real time or finite time.

Recently, \cite{stogin2020provably} showed that tensor RNNs are Turing equivalent with floating point weights and finite time. It is worth noting that Turing equivalence is a much stronger property compared to Turing completeness (TC). While Turing equivalent machines can represent all TM operations, TC, on the other hand, means that the underlying machine can compute some operations computed by the TM. Thus, simulating Turing machines by RNNs with tensor connections offers practical and computational benefits \cite{stogin2020provably,mali2020neural}, extending our understanding of the neuro-symbolic capabilities of RNN systems. Furthermore, these machines are \textbf{interpretable} by design, as rules can be easily encoded into the states of these neural networks. This is an important direction toward constructing responsible AI systems with formal guarantees. 

To this end, we proposed a new model, known as the neural state Turing machine (NSTM), which can encode the transition or production rules of a recursively enumerable grammar into the weights of an RNN. This offers multiple benefits, such as: 
1) \textbf{Computational benefit:} The number of states in the TM is equivalent to the number of weights in NSTM, 
2) \textbf{Rule insertion and interpretability:} One can add the transition rule of a recursively enumerable grammar directly into weights of an NSTM, thus making the model interpretable by design, and, 
3) \textbf{Explainability:} One can extract rules as well as the underlying state machine from a well-trained NSTM, thus making the model explainable through extraction. Crucially, by showing the equivalence between a bounded-precision NSTM and a TM that can run any algorithm, the NSTM represents a practical and computationally efficient method for symbolic processing capabilities. 
Furthermore, this equivalence can help us to more deeply understand the  computational limitations of ANNs as well as provide a pathway towards model interpretability and explainability. 
Notably, the NSTM can also be seen as a logical extension of the neural state pushdown automata \cite{mali2020neural}.

This work makes the following core contributions:
\begin{itemize}[noitemsep,nolistsep]
    \item We formally prove the Turing equivalence of a neural state Turing machine (NSTM) with bounded precision weights and time, 
    \item We extend the universal approximation theorem for ANNs to show that, with more complexity, one can avoid recurrence and still get a bounded precision ANN that is also Turing equivalent, 
    \item We formally prove that one can encode a transition table and table symbols into the weights of an NSTM -- we also show that an NSTM with six neurons can simulate a universal Turing machine, and, 
    \item We prove that one can simulate TM in real-time and that the NSTM, by design, is interpretable since the number of states in the NSTM is equivalent to a TM.
\end{itemize}
The remainder of the paper is structured in the following manner. Section \ref{sec:background} provides the background, notation, and the Turing machine (TM) definition. Section \ref{sec:bounded_nstm} defines the NSTM architecture and its equivalence with a TM. It further defines the NSTM architecture without recurrence and shows its equivalence to the TM by introducing higher-order connections that can encode a complete transition table into weighted connections. Section \ref{sec:experiments} presents experiments on grammatical inference tasks. Finally, we discuss limitations as well as the societal impact of our work.

\section{Background and Notation}
\label{sec:background}

A Turing machine (TM) is a 7-tuple \cite{hopcroft2001introduction} defined as:  $M=\langle Q,\Gamma,b,\Sigma,\delta,q_s,F \rangle$, where $Q=\{q_1(=q_s),q_2,\ldots,q_n\}$ is the set of finite states, $\Gamma =\{s_0,s_1,\ldots,s_m\}$ is the input alphabet set including $b=s_0$, $\Sigma$ is a finite set of tape symbols, $\delta: \Gamma\times Q \to \Gamma\times Q\times \{-1,1\}$ is the transition function or rule for the grammar, $q_s \in Q$ is the initial starting state, and $F \subset Q$ is the set of final states. Let the cells on the tape be indexed by the integers $1,2,3,\ldots$ according to their positions on the tape. The configuration of a TM is typically defined as a tuple consisting of the index of a cell, the symbol written in a cell, the current state of the controller, and the position or location of the read/write head. In the appendix, we provide a notation table that lists all symbols used in this work along with their definitions.

\section{Bounded Neural State Turing Machines}
\label{sec:bounded_nstm}

The design of the neural state Turing machine (NSTM) is inspired by neural state pushdown automata (NSPDA) \cite{mali2020neural} and, as a result, consists of fully-connected set of recurrent neurons labeled as \emph{state neurons} (primarily to distinguish them from neurons that function as output neurons). Introducing the concept of state neurons is important when considering tensor networks since they are composed of higher-order recurrent weights, which stands in strong contrast to first-order RNNs. Higher-order weights facilitate a direct mapping of the transition rules available in a TM grammar to the recurrent weights of NSTM itself. 

An NSTM is an artificial neural network (ANN) consisting of $n$ neurons, that are ideally equivalent to the number of states in a Turing machine. The value of neuron ${i+1}$ at time ${t+1} \in {2,3...}$ is denoted as $z^{i+1}_{t+1} \in \mathbb{Q}$, where $\mathbb{Q}$ is the set of rational numbers, and is computed by an affine transformation (or tensor multiplication) of the the values of the neurons in the current state followed by an non-linear activation function $\sigma$, formally represented as follows:
\begin{align}
    {s}_{t+1}^{i+1} &= \sigma(\Sigma_{j,k,l}W_{ijkl}^s(z_t^j,r_t^k,x_t^l)+b_s^i) \label{eqn:state}
    \\
    a_{t+1}^{i+1} &= \sigma(\Sigma_{j,k,l}W_{ijkl}^{a}(z_t^j,r_t^k,x_t^l)+b_s^i) \label{eqn:action}
    \\
    {z}_{t+1}^{i+1} &= \sigma(\Sigma_{j,k,l}W_{ijkl}^s(z_t^j,r_t^k,x_t^l) a^{t+1}_{i+1}+b_z^i) 
    \\
    {z}_{t+1}^{i+1} &= \sigma(\Sigma_{j,k,l}s^{i}_{t} a^{t+1}_{i+1}+b_z^i) 
\end{align}
where $z_{t+1} \in \mathbb{Q}$ $W^s$ and $W^{a}$ are both 4-dimensional (4D) weight tensors, i.e., the binary ``to-state'' tensor $W^s \in \{0,1\}^{J \times L \times L \times L}$ and the 4D tenary to-action tensor $W^a \in \{0,1\}^{J \times L \times L \times L}$. $s$ is the state neuron, ${a}$ is the action neuron that can be used as a tape symbol or neurons that store local information \footnote{Note that we will later show how one can manipulate the action neuron itself.}, $z$ is obtained by applying an affine transformation over $s$ and ${s}$, $b$ denotes the biases, $x$ is the input neuron, and $r$ denotes the read neuron used to determine position of tape head. For simplicity, we consider the saturated-linear function in this work, formally defined as follows:
\begin{align}
    \sigma(z) &:=
    \left\{ 
    \begin{array}{rl}
      0 & \text{if}\quad z < 0 \\ 
      z & \text{if}\quad 0 \leq z \leq 1 \\ 
      1 & \text{if}\quad x > 1 .
    \end{array}
    \right.\  \label{eqn:read}
\end{align}
Therefore, $z(t+1) \in (\mathbb{Q} \cap [0,1]^n)$ for all $t > 0$. The notation presented here is also summarized in the Appendix, for convenience. We will also show that our proof holds true even when the logistic sigmoid function is applied to the states of the NSTM. To achieve this, we will re-state the following lemma \cite{stogin2020provably}:
\begin{lemma} \label{stack_op_lem_2}
  Let $\bar{Z}$ be a state tensor with components $\bar{Z}_i\in\{0,1\}$, let $\epsilon_0<\frac{1}{2}$, and let $Z$ be a tensor satisfying:
  $$||Z-\bar{Z}||_{L^\infty} \le \epsilon_0.$$

  \noindent
  For any number $H$, let $h_H(x)=\frac{1}{1+e^{-Hx}}$ be the sigmoid function with scale $H$. Then, for all sufficiently small $\epsilon > 0$ and for all sufficiently large $H$ (the exact value depends on $\epsilon_0$ and $\epsilon$),
  $$\max_i\left|\bar{Z}_i -h_H\left(Z_i-\frac12\right)\right| \le \epsilon.$$
\end{lemma}
Note that $h_H(0)=\frac12$ and $h_H(x)$ decreases to $0$ as $Hx\rightarrow -\infty$ and increases to $1$ as $Hx\rightarrow\infty$. The scalar $H$ is a sensitivity parameter. In practice, depending on the sign of the input $x$, we assume $H$ to be a positive constant such that $h_H(x)$ is as close as desired values either $0$ or $1$. 
\begin{proof}
For simplicity, we will assume that $\bar{Z}$ to be an tensor represented in vector form.

Since $\epsilon_0<\frac12$, choose $H$ sufficiently large such that 
  $$h_H\left(-\left(\frac12-\epsilon_0\right)\right) = 1-h_H\left(\frac12-\epsilon_0\right) \le \epsilon.$$

  \noindent
  Importantly, note that we have
  \begin{align*}
     \left|\bar{Z}_i -h_H\left(Z_i-\frac12\right)\right| = \left|\bar{Z}_i -h_H\left(\bar{Z}_i-\frac12+(Z_i-\bar{Z}_i)\right)\right| 
  \end{align*}
  Now pick an arbitrary index $i$. If at that position we set $\bar{Z}_i=1$, then the following condition is obtained: 
  \begin{align*}
  \left|\bar{Z}_i-h_H\left(Z_i-\frac12\right)\right| = 1-h_H\left(\frac12+(Z_i-\bar{Z}_i)\right) \le 1-h_H\left(\frac12-\epsilon_0\right) \le \epsilon.
  \end{align*}
  
  \noindent
  However, if at position $i$ we set $\bar{Z}_i=0$, then we obtain the following condition:
  \begin{align*}
    \left|\bar{Z}_i-h_H\left(Z_i-\frac12\right)\right| = h_H\left(-\frac12+(Z_i-\bar{Z}_i)\right) 
    \le h_H\left(-\left(\frac12-\epsilon_0\right)\right)
    \le \epsilon.  
  \end{align*}
  In both cases, we observe that vector $\bar{Z}$ is close to the ideal vector Z; thus proving the lemma. 
\end{proof}

\subsection{Representing NSTMs as Turing Machines}
\label{sec:nstm_as_tm}

Recall that $M$ is a 7-tuple representing a Turing machine (TM). Given any instance of $M$, let $s_{t+1}$ be the configuration at the $t+1$ step of the TM. Given this, it is straightforward to represent the configuration $s_{t+1}$ as a characteristic tensor in the range $[0,1]$. Therefore, the configuration is ${s}_{t+1} = 1$ if and only if the ${i}^{th}$ cell contains the symbol $s_j$, the state of the controller is currently at $q_k$, and the tape head is currently pointing to the $l^{th}$ cell. Otherwise, we simply represent the configuration as $s_{t+1} = 0$.

Due to the availability of higher-order connections/weights $W^s_{ijkl}$, it can be seen that the above configuration at $s_{t+1}$ contains four pieces of information that define a transition or the next step of the TM. At index $i$, we encode the index or position of a cell. At index $j$, we encode the symbol that should be written into a cell. The state of the controller is encoded at index $k$ and, finally, index $l$ represents the position of the head. As seen by the TM's definition, we denote the information processed at index $i$ and $j$ as local information since they are dependent upon cell information while $k$ and $l$ are responsible for encoding global 
information. To show how our network can capture both local and global information, we introduce a new state known as the dead or halting state, denoted as $q_0$. This slightly modifies the definition of $Q$ which is now written as $ Q^*=Q\cup\{q_0\} $, thus extending the definition of $\delta$ to  
$\Gamma\times Q^* \to \Gamma\times Q^*\times \{-1,0,1\}$ which is represented as:
\[
	\delta(s,q)=
        \begin{cases}
		\delta(s,q), &\text{ if }q\in Q\\
		(s,q_0,0),   &\text{ if }q=q_0 \mbox{.}
	\end{cases}
\]
Now one needs to represent $\delta=(\delta_1,\delta_2,\delta_3)$, where $\delta_1: \Gamma\times Q^* \to \Gamma $, $\delta_2: \Gamma\times Q^* \to Q^*$, and $\delta_3: \Gamma\times Q \to \{-1,0,1\}$ can also be modelled using the hyperbolic tangent (tanh) activation function. Further note that we denote $[k]=\{0,1,2\ldots,k\}$. 

In this work, we treat an action neuron as a tensor that stores the characteristics or evolution of the Turing machine. Thus, $a_{t+1}$ can be represented as a tensor combination of the current and future states, when using the logistic sigmoid or a saturated function $a_{t+1} = \sigma(a_t)$. Another way of looking at this tensor is by using an $8^{th}$ order tensor with a threshold logic gate such that $a_{t+1}=(Wa^{i_t,j_t,k_t,l_t}_{i_{t+1}, j_{t+1}, k_{t+1}, l_{t+1}})$, where $i_t, i_{t+1} \in \mathbb{N}^+, l_t, l_{t+1} \in \mathbb{N}^+, j_t, j_{t+1} \in [m]$ and $k_t, k_{t+1} \in [n]$, $t$ represents the current configuration and $t+1$ represents future or upcoming configuration(s). Now this representation can model the evolution of a TM with various activation functions. In the appendix, we will also show that, under a fixed point analysis, the NSTM is guaranteed to find stable solutions when the transition rules are well-defined or are represented by the recurrent weights. 

The configuration (or action) neuron is $a_{t+1} = 1$ if it satisfies two conditions. First, $i_t \neq l_t, i_{t+1} = i_t, j_{t+1} = j_t, k_{t+1}=0, k_t \neq 0, l_{t+1} = l_{t}$ and, second,  $i_{t} = l_{t}, i_{t+1}=i_t, j_{t+1}= \delta_1(j_{t}, k_t), k_{t+1}=\delta_2(j_t,k_t), l_{t+1}=l_t+\delta_3(j_t,k_t)$ or else $a_{t+1} = 0$. The first condition checks whether the cell is in inactive states or not based on positions of the head while the second condition is responsible for checking if the current cell is active. However, this only satisfies a local information criterion, since, based on the second condition, when $i_t = l_t$, the current cell is in an active state and a future/upcoming state can be easily determined. As a result, in the scenario when the state is inactive such that $i_t \neq l_t$, one can only determine the cell index and the symbol written on the tape. In these scenarios, when global information and the full transition to the next state are not fully determined, we set the next state to be equivalent to halting state $q_0$. In other words, we set the halting state to be equivalent to the next state when $a_{t+1}=1$ and $k_t \neq 0$. One key idea is that the halting state only interferes when global information is unavailable. This state eventually vanishes in the next step, ensuring that it is not involved in further calculations. The significance of a halting state is that it can easily deal with illegal configurations, which is critical for dealing with imbalanced datasets and adversarial samples. 

Next, the state neurons $s$ and action neurons $a$ are combined and an affine transformation is applied to them in order to create binary operation of tensors $z$, which models a TM's evolution. This is represented as:
\begin{align}
    \sigma(s_{t} a_{t+1}) \triangleq \sum_{i_t,j_t,k_t,l_t} (Ws_{i_t,j_t,k_t,l_t}) a_{t+1}.
\end{align}
If we replace $\sigma$ with a threshold activation function and make $a_{t+1}$ equivalent to an $8^{th}$ order tensor, as shown earlier, we can use \textbf{Einstein notation} \cite{einstein1916grundlage} to model two tensors such that: 
\begin{align}
\begin{split}
    (Ws_{i_t, j_t, k_t, l_t})(Wa^{i_t,j_t,k_t,l_t}_{i_{t+1}, j_{t+1}, k_{t+1}, l_{t+1}}) \\ \triangleq 
  \sum_{i_t,j_t,k_t,l_t} (Ws_{i_t,j_t,k_t,l_t})(Wa^{i_t,j_t,k_t,l_t}_{i_{t+1}, j_{t+1}, k_{t+1}, l_{t+1}}) .
\end{split}
\end{align}
Prior work \cite{einstein1916foundation} has shown that these summations are well-defined and obey both cumulative and distributive laws across tensors. 

Now let us define the product of two tensors and describe further how to model the local and global information of a TM with an NSTM. We show that this definition holds true for both the sigmoid activation as well as a threshold logic gate (when applied over an $4^{th}$ order tensor connection). 
\begin{definition}\label{def:Type1prod}
Let $s_t = $ $Ws_{i_t,j_t,k_t,l_t}$ or $\Sigma_{j,k,l}W_{ijkl}^s(z_t^j,r_t^k,x_t^l)$, 
$a^l_{t+1} = \sigma( Wa^l_{i_t,j_t,k_t,l_t} ) \; \mbox{ or } \;  \sigma( \Sigma_{j,k,l}Wa_{ijkl}^l(z_t^j,r_t^k,x_t^l))$ \; $\mbox{ and } \; $  $a^g_{t+1} =\sigma( Wa^g_{i_t,j_t,k_t,l_t} ) \; \mbox{ or  } \; \sigma( \Sigma_{j,k,l}Wa_{ijkl}^g(z_t^j,r_t^k,x_t^l))$ be a tensor of order 4. We define the Type 1 tensor product of $s_t$ and $a_{t+1}$ to be $z_t = s_t \times_1 a_{t+1}$, where $\sigma$ or $h_H$ is a sigmoid activation function and $z_t$ is a tensor of $4^{th}$ order such that: 
\begin{align*}
    z_t = \sigma \Bigg( \bigg( \sum_{k_{t+1}, l_{t+1}}s_t Wa^l_{i_t, j_t, k_{t+1}, l_{t+1}} \bigg)
    \bigg( \sum_{i_{t+1}, j_{t+1}}s_t Wa^g_{i_{t+1}, j_{t+1}, k_{t}, l_{t}} \bigg) \Bigg) .
\end{align*}
Similarly, we can show that the evolution of an NSTM closely follows Einstein notation even with an $8^{th}$ order tensor and, furthermore, closely follows the associative property. 
As such, the same tensor stores current and future information. Therefore, the above equations can be written as: 
\begin{align*}
    z_t = \sigma \Bigg( \bigg( \sum_{k_{t+1}, l_{t+1}} Ws_{i_t, j_t, k_t, l_t} Wa^{i_t,j_t,k_t,l_t}_{i_{t}, j_{t}, k_{t+1}, l_{t+1}} \bigg)  \bigg( \sum_{i_{t+1}, j_{t+1}} Ws_{i_t, j_t, k_t, l_t} Wa^{i_t,j_t,k_t,l_t}_{i_{t+1}, j_{t+1}, k_{t}, l_{t}} \bigg) \Bigg). \nonumber
\end{align*}
\end{definition}
As we can see from Definition \ref{def:Type1prod}, the first summation shows the interaction between the state and action neurons, which captures the local information,  while the second summation captures the global information. When the state neuron $s_t$ is represented as a $4^{th}$ order tensor such that $k \in [n]$, then $s|_{k \ne k_i}$ indicates the tensor obtained from $s_t$ by deleting the $k$ index at $k_i$. For instance, $s|_{k \ne 0}$ denotes the tensor obtained by deleting entries from $s_t$ with an index of 0.

This next theorem is essential and shows how the NSTM construction is close to that of a TM.

\begin{theorem}\label{thm:product_and_DTM}
Let $M$ be a TM such that $z_t$ represents the $t^{th}$ configuration of $M$ at a given instance. Let $a_{t+1}$ be the evolution tensor in the NSTM that stores the characteristic of $M$ by interacting with a nonlinear activation  ($\sigma$). Suppose that the state neurons of NSTM $s_t$ are stored in a $4^{th}$ order tensor such that $s_t|_{k\neq 0} = s(z_t)$, then let us define $s_{t+1} = s_t \times_1 a_{t+1}$ for $t=1,2, \ldots.$. Show that $s(z_t) = s_t|_{k \neq 0}$. 
\end{theorem}
\begin{proof}
We are interested in showing that the interaction between the state neurons $s_t$ and the product of state and action neurons $z_t$ at the current time step and the next time step lie within a similar range. We will follow proof by induction, where we only need to prove the following conditions:  if $s_t|_{k \ne 0} = s(z_t)$, then at the next time-step $s(z_{t+1}) = s_{t+1}|_{k \ne 0}$. It is crucial to note that this proof is valid for both threshold activation functions and sigmoidal functions. In the case of the sigmoid, we will show $s(z_{t+1}) = s_{t+1}|_{k \ne 0} + \epsilon$, where $\epsilon$ is as small as possible and the solution lies within a fixed point threshold. 
We know $s_t = \sigma(Ws_{i_t, j_t, k_t, l_t})$ for all $t \in \N^+$. Here we only need to show that, for $\forall k\neq 0$, the state neuron at the next time-step $s^{(t+1)}_{i,j,k,l} \neq 0$ if and only if both summations equate or are equal to $1$. 
For the condition when $(\left(\sum_{k_{t+1}, l_{t+1}}s_t Wa^l_{i_t, j_t, k_{t+1}, l_{t+1}} \right) \neq 0$, then by definition of $a_t$, the local tensor $a^l_{t+1} \neq 0$ when $i_{t+1} = i$ and $k_t \neq 0$. Since $s_t|k neq 0$ stores the configuration of the TM ($M$) at $t$ and can also be derived from the NSTM final state $z_t$, there then exists only one condition where $s_t = 1$ when $i_{t} = i_{t+1}$ and $k /neq 0$, thus making $(\left(\sum_{k_{t+1}, l_{t+1}}s_t Wa^l_{i_t, j_t, k_{t+1}, l_{t+1}} \right) = 1$. 

Based on the above formulation, which creates a binary and sparse representation, it is easy to determine that the non-zero term will indicate that the $i^{th}$ cell contains the symbol $s_j$ in $z_{t+1}$. In a similar fashion, we can retrieve the global information by checking that $\left(\sum_{i_{t+1}, j_{t+1}}s_t Wa^g_{i_{t+1}, j_{t+1}, k_{t}, l_{t}} \right) = 1$, if and only if the final state of the NSTM is in $z_{t+1}$ and the current state based on the transition rule is in $q_{k_t}$ and, furthermore, that the position of the head is in the $l^{th}_t$ cell. This completes the proof and shows that the NSTM closely follows the transition evolution of a TM $M$. The above proof also shows that the NSTM can compute any TM grammar in real-time.
\end{proof}

Next, we will derive the lower bound of an NSTM that recognizes a small universal Turing machine (UTM). Prior work has proposed four small UTMs \cite{neary2007four}, where one such configuration contains $6$ states and $4$ symbols that can simulate any Turing machine in time O($T^6$) (where $T$ represents number of steps that the TM requires to compute a final output). As shown there, the $UTM_{6,4}$ or $UTM_{7,3}$ is a TM. We now use Theorem \ref{thm:product_and_DTM} to simulate a UTM, thus showing the minimal number of neurons required by a NSTM to recognize any TM grammar. 
\begin{corollary}
There exists a $6$ neuron unbounded-precision NSTM that can simulate any TM in O($T^6$), where $T$ is the total number of turns or steps required by the TM to compute the final output.
\end{corollary}
Assume that $s_t$ inserts a transition rule as described by our construction and also models the evolution similar to $a_t$. However, when we add an additional state (which is termed the ``halting'' state), with characteristics of a TM or an action neuron, we obtain bounded precision. The rest of the construction follows from Section \ref{thm:product_and_DTM}. 
\begin{corollary}
There exists a $13$ neuron bounded NSTM that can simulate any TM in O($T^6$), where $T$ is the total number of turns or steps required by the Turing machine to compute a final output.
\end{corollary}
The only difference here is the addition of a halting state and that the evolution of Turing machine is modeled using an additional tensor $a_t$. It should be noted that prior results \cite{chung2021turing,perez2019turing} do not function in real-time and, even with $O(T^6)$ simulation time, the number of neurons required by a bounded RNN with growing memory is $54$ \cite{chung2021turing} as opposed to the $13$ required by our construction. However, \cite{stogin2020provably} works in real-time but uses two stacks.

\section{Turing Completeness of an NSTM with an n-th Order Tensor }
\label{sec:turing_completeness}

Next, we show how the evolution of an NSTM closely resembles a Markov process and that a  UTM can be derived from an NSTM using linear activations. In other words, the evolution of $UTM = NSTM$ with a linear activation function and $UTM = NSTM + \epsilon$ with a sigmoid function ($h_H$), for a sufficiently larger value of $H$ along with an $\epsilon$ that is as small as possible, as shown in Lemma \ref{stack_op_lem_2}. 
We use transition matrices to prove that the NSTM evolution is close to a Markov process. In an earlier description, we showed how a state neuron $s_t$ models the configurations of $M$ by interacting with the states and the tape based on a defined transition rule. The transition between the current and next state of a TM is then modeled using an action neuron $a_{t+1}$. We also define the final state neuron $z_t$ responsible for modelling one step of evolution of a TM using a Type 1 tensor. Here, we show that our results can be generalized across the order of tensor connections and can also be modelled using Type 2 tensor products. Since we know that ANNs behave similarly to Markov processes \cite{puterman2014markov,lin2016critical}, representing the complete state transition helps show us show that the dynamics of a TM are similar to an NSTM and that the state-transition and higher-order connections are critical for modeling these dynamics. 

Next we know ANNs with infinite precision and time are universal and can approximate any function using only a single layer of neurons using the sigmoid activation function \cite{cybenko1989approximation}. In this context, we now provide a tighter bound showing that \textbf{ANNs, when augmented with memory and higher-order connections, can also approximate any function crucially with finite precision and time}. Desirably, we provide a new bound for ANNs that can exploit the desirable properties of neural transformers, since such models do take into account the entire state-transition history using a combination of positional encodings and attention.  
Let $a_t$ represent an action neuron with recurrence. If we remove the recurrence denoted as ($A^r$) and encode the entire (input) history, then $A^r = ({a^r}^{i_1j_1,k_1l_1;\ldots ;i_pj_p,k_pl_p}_{i_{p+1}j_{p+1},k_{p+1}l_{p+1}})$. Similarly, we can represent the final state neuron of an NSTM without recurrence as ($Z^r$ = $({z^r}^{i_1j_1,k_1l_1;\ldots ;i_qj_q,k_ql_q}_{i_{q+1}j_{q+1},k_{q+1}l_{q+1}})$ ), where $p$ and $q$ represent the order of the tensor. Here, we will show how an NSTM without recurrence also obeys the associative law and can model any class of TM: 
\begin{table}
    \centering
    \begin{tabular}{l|r||r||r|r}
    \multicolumn{1}{l}{}&\multicolumn{1}{|c}{\begin{tabular}[x]{@{}c@{}}\textbf{$Train$}\\\end{tabular}}&\multicolumn{1}{|c}{\begin{tabular}[x]{@{}c@{}}\textbf{$Test$}\\\end{tabular}}&\multicolumn{2}{|c}{\begin{tabular}[x]{@{}c@{}}\textbf{$Long Strings$}\\\end{tabular}}\tabularnewline
    \multicolumn{1}{l}{\textbf{RNN Models}}&\multicolumn{1}{|c}{\begin{tabular}[x]{@{}c@{}}\textbf{Mean}\\\end{tabular}}&\multicolumn{1}{|c}{\begin{tabular}[x]{@{}c@{}}\textbf{Mean}\\\end{tabular}}&\multicolumn{1}{|c}{\begin{tabular}[x]{@{}c@{}}\textbf{n=500}\\\end{tabular}}&\multicolumn{1}{|c}{\begin{tabular}[x]{@{}c@{}}\textbf{n=1000}\\\end{tabular}}\tabularnewline
\hline
        \textit{LSTM} & $100$ & $98$ & $62$ & $35.5$  \tabularnewline
        \textit{Stack-RNN} & $100$ & $98.5$ & $80$ & $49.00$  \tabularnewline
        \textit{NTM} & $100$ & $99.99$ & $82.5$ & $65.00$   \tabularnewline
        \textit{Transformer} & $100$ & $92.5$ & $36.5$ & $4.00$
        \tabularnewline
        \textit{NSTM (\textbf{ours})} & $100$ & $99.99$ & $99.50$ & $99.4$  \tabularnewline
    \end{tabular}%
    \caption{Percentage of correctly classified strings of RNNs trained on the $D_4$ language (over $3$ trials). We report the mean accuracy for each model. The test set contained a mix of short \& long samples of length up to $T = 120$ with both positive and negative samples. We also report the mean accuracy of models tested on longer strings ($n=500$ \& $n=1000$)}
    \label{D4performance}
\end{table}
\begin{definition}
    Let action neurons, without recurrence, $A^r$ and a threshold activation function be a tensor of order $4p+1$ and let the final state neuron without recurrence ($Z^r$) be a tensor of order $4q+1$ such that $p,q \in \mathbb{N}^+$. Therefore, the interaction between the current action neuron and state neuron f($A^r \times_2 Z^r$) is represented as a Type 2 tensor and can be defined using a tensor without recurrence of order $4(2pq+1)$ such that $Z^f= ({z^f}^{i_1j_1,k_1l_1;\ldots ;i_{2pq}j_{2pq},k_{2pq}l_{2pq}}_{i_{2pq+1}j_{2pq+1},k_{2pq+1}l_{2pq+1}})$ and indices $I^{''} =(i_1''j_1'',k_1''l_1'';\ldots;i_q''j_q'',k_q''l_q'') $. Thus, the complete composition of this function can be shown to be:
\begin{align*}
Z^{f} =  &F(\sum_{I^{''}}  s^{r{i_1j_1,k_1l_1; \ldots ;i_pj_p,k_pl_p}}_{{i_{1}'j_{1}'},k_{1}''l_{1}''} 
&\cdot a^{r{i_{p+1}j_{p+1},k_{p+1}l_{p+1}; \ldots ;i_{2p}j_{2p},k_{2p}l_{2p}}}_{i_{1}''j_{1}'',{k_{1}'l_{1}'}} 
&\cdot a^{r{i_{2p+1}j_{2p+1},k_{2p+1}l_{2p+1}; 
\ldots  ;i_{3p}j_{3p},k_{3p}l_{3p}}}_{{i_{2}'j_{2}'},k_{2}''l_{2}''}
\nonumber \\ &\cdot a^{r{i_{3p+1}j_{3p+1},k_{3p+1}l_{3p+1}; 
\ldots ;i_{4p}j_{4p},k_{4p}l_{4p}}}_{i_{2}''j_{2}'',{k_{2}'l_{2}'}} &\cdot {z^f}^{{i_1'j_1',k_1'l_1';\ldots;i_q'j_q',k_q'l_q'}}_{{i_{2pq+1}j_{2pq+1},k_{2pq+1}l_{2pq+1}}}). 
\end{align*}
where $F$ is a threshold logic gate and, even when using the logistic sigmoid, for sufficiently larger values of $H$, we can find that the neural activities converge to $0$ or $1$. Thus, using the sigmoid ($\sigma$) or ($h_H$) function, for sufficiently larger values of $H$, without loss of generality and minimum error $\epsilon$, we can represent the complete, differentiable transition as follows:
\begin{align*}
Z^f = &h_H (\sum_{I^{''}}  s^{r{i_1j_1,k_1l_1;\ldots ;i_pj_p,k_pl_p}}_{{i_{1}'j_{1}'},{k_{1}''l_{1}''}} &\cdot a^{ r{\mathbf{t1}} } 
&\cdot a^{ r{\mathbf{t2}} }
&\cdot a^{ r{\mathbf{t3}} }
\quad   
&\cdots \cdots a^{r{\mathbf{tn}}} \nonumber \\  
&\cdot z^{f{i_1'j_1',k_1'l_1';\ldots;i_q'j_q',k_q'l_q'}}_{ {i_{2pq+1}j_{2pq+1}},{k_{2pq+1}l_{2pq+1}} }) 
+ \epsilon  
\end{align*}
where we further define:
\begin{align*}
\omega_1 &= {i_{p+1}j_{p+1},k_{p+1}l_{p+1};\ldots ;i_{2p}j_{2p},k_{2p}l_{2p}} \\
\omega_2 &= {i_{2p+1}j_{2p+1},k_{2p+1}l_{2p+1}; \ldots ;i_{3p}j_{3p},k_{3p}l_{3p}} \\
\omega_3 &= {i_{3p+1}j_{3p+1},k_{3p+1}l_{3p+1}; 
\ldots ;i_{4p}j_{4p},k_{4p}l_{4p}} \\
\delta_{e} &= (2q-1)p+1 \\
\delta_{e+1} &= 2pq \\
\omega_m &= {i_{\delta_{e}}j_{\delta_{e}},k_{\delta_{e}}l_{\delta_{e} } \ldots;i_{\delta_{e+1}}j_{\delta_{e+1}},k_{\delta_{e+1}}l_{\delta_{e+1}}} \\
a^{r{\mathbf{t1}}} &= a^{r\omega_1}_{i_{1}''j_{1}'',{k_{1}'l_{1}'}} \\ 
a^{r{\mathbf{t2}}} &= a^{r\omega_2}_{{i_{2}'j_{2}'},k_{2}''l_{2}''} \\ 
a^{r{\mathbf{t3}}} &= a^{r\omega_3}_{i_{2}''j_{2}'',{k_{2}'l_{2}'}} \\
a^{r{\mathbf{tn}}} &= a^{r \omega_m}_{i_q''j_q'',k_q'l_q'}  
\end{align*}
noting that $\mathbf{t}$ captures the evolution of $p$ and $\mathbf{n}$ captures the evolution of $q$.
\end{definition}

\begin{definition}
    Assume that the state neurons of an NSTM without any recurrence are represented as $S^r = \sigma(Ws^{r}_{i,j,k,l}z_jr_k,x_l)$, representing the configuration of TM $M$. Let $A^r$ be the transition function of order $4(p+1)$, then the Type 1 product ($S^r \times_1 A^r$ to get the final state neurons ($Z^r$) of order $4$ is represented as:
    \begin{align*}    
         S^r = \sigma \Bigg(\bigg(\sum_{k',l'} S^r_i\cdots S^r_p Wa^{r^{i_1j_1,k_1l_1;\ldots ;i_pj_p,k_pl_p}}_ {ij,k'l'}\bigg) \bigg(\sum_{i',j'} S^r_i\cdots S^r_p Wa^{r^{i_1j_1,k_1l_1;\ldots ;i_pj_p,k_pl_p}}_ {i'j',kl}\bigg) \Bigg)   
    \end{align*}
  where $\sigma$ is the activation function and $Wa^r$ is the tensor weights for action neuron $A^r$. With this definition, it is possible to show equivalence between an NSTM with a TM without recurrence but with an increased complexity by introducing an n-order tensor.
\end{definition}

Based on the following definitions, we can easily verify the following proposition. Recall that the small universal Turing machine (UTM) based on \cite{neary2007four} consists of $6$ states with $4$ symbols. Therefore, we show that a $6$ neuron higher-order NSTM, without recurrence, can model a TM with bounded precision and weights:
\begin{proposition}
    Let the states of the NSTM without recurrence be represented by $S^r$. An action neuron, which also models the transition of a TM, is represented as $A^r$ and the final state neuron is represented as $Z^r$. Then, based on the Type 1 and Type 2 definitions presented before, we see that: 
    $\sigma((S^r \times_1 A^r) \times_1 Z^r )$ = $\sigma(S^r \times_1 (A^r \times_2 Z^r))$. Thus, a Type 2 product also obeys the associative law, proving that the NSTM without recurrence reconstruction is close to an ideal TM $M$ and thus follows the evolution of a Markov process.
\end{proposition}

\begin{corollary}
There exists a $12$ neuron bounded NSTM without recurrence and $25^{th}$ order tensor weights that can simulate any Turing machine in O($T^6$), where $T$ is the total number of turns or steps required by the TM to compute a final output. 
\end{corollary}
The above corollary shows that there exists a \textbf{bounded weight and precision ANN with only feedforward connections that can model any class of TM grammar}.

\section{Experimental Setup and Results}
\label{sec:experiments}

To test our hypothesis, we experimented with complex grammars, such as the Dyck languages, and tested various models on longer strings. The Dyck languages are defined in the following manner: 
Let $\Sigma = {[,]}$ be the alphabet consisting of symbols '[' and ']' and let $\Sigma^*$ denote its Kleene closure. A Dyck language is then formally defined as:
\begin{align}
&\{n \in \Sigma^*  | \mbox{ all prefixes of } n 
&\mbox{ contain no more of symbol ']' than symbol '['  and }  
&cnt(\mbox{'['}, n) = cnt(\mbox{']'}, n)  \}
\end{align}
where $cnt([\text{symbol}],n)$ is the frequency of [symbol] in $n$. We tested our model on various Dyck languages, such as $D_2$, $D_3$, and $D_4$. 
We created a dataset containing $5000$ training samples (of length $T$ less than or equal to $52$), $500$ for validation ($20 < T \leq 70$), and $3000$ samples for testing ($52 <  T \leq 120)$. We also created two separate test sets with $1000$ samples of each with lengths in the ranges ($120 <  T \leq 500)$ and ($500 <  T \leq 1000)$, respectively. 

\textbf{Baselines:} We compare our NSTM with modern-day ANNs that have and do not have memory, such as the long short-term memory (LSTM) \cite{hochreiter1997long}, stack-RNN \cite{joulin2015inferring}, the neural Turing machine (NTM) \cite{graves2014neural}, and transformers \cite{vaswani2017attention}. All baseline models are trained using backpropagation of errors (backprop) and backprop through time (BPTT). To select the optimal hyperparameters, we perform a grid search on batch sizes in the range [$32-128$], hidden unit sizes [$64-1024$], and use a learning rate [$0.1 - 1e-4$]. All models are trained with stochastic gradient descent and the learning rate is reduced to half using patience scheduling. All models are trained over $400$ epochs with a mean square error cost function (early stopping is used to obtain the optimal model). The best (model) settings are reported in the Appendix. 

In contrast to the baselines, the NSTM is trained using real-time recurrent learning (RTRL) \cite{williams1989rtrl} with only $8$ neurons; we report the performance of the NSTM with different numbers of hidden units as well as the average number of epochs required to converge. 
We also report the average accuracy of all models and provide a detailed analysis in the appendix. 
We also report the NSTM's performance on the $D_2$ and $D_3$ grammars in the appendix. To have a fair comparison with the transformer, we adapt the transformer model proposed in \cite{attn_2} with four attention heads, referred to as $SA^+$, which was shown to perform better than a vanilla transformer when recognizing Dyck languages. Our experiments with the NSTM show that it consistently performs better than the other neural models. Even though it is challenging to train memory-augmented NNs in general, it is seen that, in order to learn complex languages, simply increasing the number of parameters in any ANN is insufficient; one needs memory access to store the intricate patterns. All memory-augmented structures achieve better performance compared to the LSTM and even transformers. Similar findings are also reported in other works \cite{mali2020neural, suzgun2019memoryaugmented, deletang2022neural}; however, most of these prior works focus on first-order ANNs with and without memory. 

In the appendix, we report baseline model results and NSTM parameters that performed well (also reporting ranges used for the parameter sweep/grid-search). We provide a detailed analysis of our results in the appendix.

\textbf{Results:} Note that these experiments are not meant to achieve state-of-the-art performance but rather to showcase the benefits offered by tensor connections both in terms of interpretability and performance. In Table \ref{D4performance}, we observe that the NSTM with only $8$ neurons outperforms all baselines by a wide margin on longer strings. Note that the NSTM is trained from scratch instead of manually programming the weights. It is interesting to see that, in Table \ref{D4performance}, modern-day ANNs, with even more parameters, struggle to recognize the $D_4$ grammar when tested on longer strings. We also conducted an ablation study and reported the average number of epochs required by each model to generalize well on the validation set. We report the results for other settings as well as the results for other grammars in the Appendix. 

\section{Conclusion}
\label{sec:conclusion}

In this work, we introduced a new class of formal models called the Neural State Turing Machine (NSTM), which is interpretable and capable of modeling any Turing machine. We show the algebraic representation of a tensor neural network for recognizing recursively enumerable languages and prove that the Universal Turing completeness of a $6$-bounded state neuron recurrent neural network (RNN), the smallest Turing complete RNN to date, of any order. Prior results \cite{chung2021turing} have shown that for a $6$-state, $4$-symbol Universal Turing Machine (UTM), there exists a $40$-neuron unbounded precision RNN and a $54$-neuron bounded precision RNN augmented with two stack-like memory. Our approach vastly improves upon these prior works by using fewer neurons and introducing higher-order/tensor synaptic weights. Our model is both interpretable and efficient, making it a promising avenue for developing more trustworthy and reliable machine-learning applications. By leveraging tensor synaptic connections, the NSTM is able to model complex rules in real-time and with finite-precision weights, which is a significant advantage over prior approaches that assume unbounded computation time and rational weights. These results have implications for developing more powerful and \textbf{interpretable neural models} and serve as a plausible candidate in the advancing field of responsible AI. Future work will be dedicated to carefully analyzing the capability of stably extracting state machines from these networks and designing stable training algorithms to train such systems. 

\paragraph{Limitations:} This work is primarily theoretical and, due to the higher-order synapses, the proposed NSTM is computationally expensive and challenging to scale. Second, its training depends on complex algorithms, e.g., RTRL. Despite these challenges, NSTM is interpretable, requiring fewer neurons and data, and desirably generalizes to longer sequences. These features are important to lower carbon footprint induced by the training of large LLMs and, thus, more development in this area would lead to the creation of environmentally-friendlier, responsible AI systems. Details  with respect to societal impact can be found in the appendix.

\bibliographystyle{acm}
\bibliography{custom}
\newpage
\appendix
\section{Appendix A: Experimental Design and Ablation Study}
As noted in the main paper, the NSTM is trained using real-time recurrent learning (RTRL) \cite{williams1989rtrl} with only $8$ neurons; we report the performance of the NSTM with different numbers of hidden units as well as the average number of epochs required to converge. 
We also report the average accuracy of all models here and provide a detailed analysis. In Table \ref{best_settings}, we report the best settings for each model.  
We also report the NSTM's performance on the $D_2$ and $D_3$ grammars in Table \ref{D2performance} and \ref{D3performance}, respectively. To have a fair comparison with the transformer, we adapt the transformer model proposed in \cite{attn_2} with four attention heads, referred to as $SA^+$, which was shown to perform better than a vanilla transformer model in recognizing Dyck languages. Our experiments with the NSTM show that it constantly performs better than the other artificial neural networks (ANNs) with and without memory. Even though it is challenging to train memory-augmented NNs in general, it is seen that to learn complex languages, simply increasing the number of parameters in any ANN is insufficient; one needs memory access to store the intricate patterns. All memory-augmented structures achieve better performance compared to the LSTM and even transformers. Similar findings are also reported in other works \cite{mali2020neural, suzgun2019memoryaugmented, deletang2022neural}; however, most of these prior works focus on first-order ANNs with and without memory. 

We also conduct an ablation study to show how the NSTM works with different hidden layer sizes \ref{D4performance_neurons_nstm}. As noted earlier, it is challenging to train the NSTM with a large number of neurons using RTRL. However, despite having less number of neurons, it is evident from the theoretical and empirical findings that NSTM can generalize better compared to other models. To align theory with practical findings, we only conducted experiments with a small number of neurons.  
Table \ref{average_epochs} reports each model's average number of epochs required to achieve its best validation accuracy. We define this as a convergence criterion; it is worth noting that the transformer-based architecture never reaches $100$\% on the validation set since the positional embedding was never encoded with information related to longer strings. However, the transformer does reach $100$\% on the training set, but this can only be achieved by increasing the number of parameters.
All experiments in this work are performed using Tensorflow 2.2 framework on an Intel Xeon server with a maximum clock speed of $3.5$GHz with four $2080$Ti GPUs with $192$GB RAM. 

\begin{table}
\caption{Percentage of correctly classified strings of RNNs trained on the $D_4$ language (over $3$ trials). We report the change in model performance as the number of hidden units ($H$) are changed. We report each model's mean accuracy on long strings ($n=1000$). ``N/A'' indicates that, for the NSTM, we did not analyze other settings of $H$.}
    \centering
    \begin{tabular}{l|r||r||r|r}
    \multicolumn{1}{l}{}&\multicolumn{1}{|c}{\begin{tabular}[x]{@{}c@{}}\textbf{}\\\end{tabular}}&\multicolumn{1}{|c}{\begin{tabular}[x]{@{}c@{}}\textbf{}\\\end{tabular}}&\multicolumn{2}{|c}{\begin{tabular}[x]{@{}c@{}}\textbf{}\\\end{tabular}}\tabularnewline
    \multicolumn{1}{l}{\textbf{RNN Models}}&\multicolumn{1}{|c}{\begin{tabular}[x]{@{}c@{}}\textbf{H=8}\\\end{tabular}}&\multicolumn{1}{|c}{\begin{tabular}[x]{@{}c@{}}\textbf{H=128}\\\end{tabular}}&\multicolumn{1}{|c}{\begin{tabular}[x]{@{}c@{}}\textbf{H=512}\\\end{tabular}}&\multicolumn{1}{|c}{\begin{tabular}[x]{@{}c@{}}\textbf{H=1024}\\\end{tabular}}\tabularnewline
\hline
        \textit{LSTM} & $10.0$ & $25.00$ & $30.5$ & $35.5$  \tabularnewline
        \textit{Stack-RNN} & $25.0$ & $36.00$ & $49.0$ & $38.50$  \tabularnewline
        \textit{NTM} & $26.50$ & $45.00$ & $65.0$ & $60.00$   \tabularnewline
        \textit{Transformer} & $2.00$ & $3.0$ & $3.5$ & $4.00$
        \tabularnewline
        \textit{NSTM (\textbf{ours})} & $99.4$ & $N/A$ & $N/A$ & $N/A$  \tabularnewline
    \end{tabular}%
    
    \label{D4performance_neurons}
\end{table}
\begin{table}[htb!]
\caption{Percentage of correctly classified strings of RNNs trained on the $D_2$ language (across $3$ trials). We report the mean accuracy for each model. The test set contained a mix of short and long samples of length up to $T = 120$ with both positive and negative samples. We also report the mean accuracy of all of the models tested on longer strings ($n=500$ and $n=1000$)}
\label{D2performance}
    \centering
    \begin{tabular}{l|r||r||r|r}
    \multicolumn{1}{l}{}&\multicolumn{1}{|c}{\begin{tabular}[x]{@{}c@{}}\textbf{$Train$}\\\end{tabular}}&\multicolumn{1}{|c}{\begin{tabular}[x]{@{}c@{}}\textbf{$Test$}\\\end{tabular}}&\multicolumn{2}{|c}{\begin{tabular}[x]{@{}c@{}}\textbf{$Long Strings$}\\\end{tabular}}\tabularnewline
    \multicolumn{1}{l}{\textbf{RNN Models}}&\multicolumn{1}{|c}{\begin{tabular}[x]{@{}c@{}}\textbf{Mean}\\\end{tabular}}&\multicolumn{1}{|c}{\begin{tabular}[x]{@{}c@{}}\textbf{Mean}\\\end{tabular}}&\multicolumn{1}{|c}{\begin{tabular}[x]{@{}c@{}}\textbf{n=500}\\\end{tabular}}&\multicolumn{1}{|c}{\begin{tabular}[x]{@{}c@{}}\textbf{n=1000}\\\end{tabular}}\tabularnewline
\hline
        \textit{LSTM} & $100$ & $76.50$ & $50$ & $30.00$  \tabularnewline
        \textit{Stack-RNN} & $100$ & $99$ & $84.50$ & $57.00$  \tabularnewline
        \textit{NTM} & $100$ & $99.99$ & $89.5$ & $73.99$   \tabularnewline
        \textit{Transformer} & $100$ & $60.50$ & $32.50$ & $7.50$
        \tabularnewline
        \textit{NSTM (\textbf{ours})} & $100$ & $99.99$ & $98.50$ & $98.00$  \tabularnewline
    \end{tabular}%
\end{table}

\begin{table}
\caption{Percentage of correctly classified strings of RNNs trained on the $D_3$ language (across $3$ trials). We report the mean accuracy for each model. The test set contained a mix of short and long samples of length up to $T = 120$ with both positive and negative samples. We also report the mean accuracy of all of the models tested on longer strings ($n=500$ and $n=1000$)}
\label{D3performance}
    \centering
    \begin{tabular}{l|r||r||r|r}
    \multicolumn{1}{l}{}&\multicolumn{1}{|c}{\begin{tabular}[x]{@{}c@{}}\textbf{$Train$}\\\end{tabular}}&\multicolumn{1}{|c}{\begin{tabular}[x]{@{}c@{}}\textbf{$Test$}\\\end{tabular}}&\multicolumn{2}{|c}{\begin{tabular}[x]{@{}c@{}}\textbf{$Long Strings$}\\\end{tabular}}\tabularnewline
    \multicolumn{1}{l}{\textbf{RNN Models}}&\multicolumn{1}{|c}{\begin{tabular}[x]{@{}c@{}}\textbf{Mean}\\\end{tabular}}&\multicolumn{1}{|c}{\begin{tabular}[x]{@{}c@{}}\textbf{Mean}\\\end{tabular}}&\multicolumn{1}{|c}{\begin{tabular}[x]{@{}c@{}}\textbf{n=500}\\\end{tabular}}&\multicolumn{1}{|c}{\begin{tabular}[x]{@{}c@{}}\textbf{n=1000}\\\end{tabular}}\tabularnewline
\hline
        \textit{LSTM} & $100$ & $75.00$ & $60$ & $40.5$  \tabularnewline
        \textit{Stack-RNN} & $100$ & $84.50$ & $59.00$ & $43.99$  \tabularnewline
        \textit{NTM} & $100$ & $94.50$ & $80.5$ & $70.00$   \tabularnewline
        \textit{Transformer} & $99.99$ & $66.00$ & $29.5$ & $5.50$
        \tabularnewline
        \textit{NSTM (\textbf{ours})} & $100$ & $99.99$ & $98.50$ & $95.00$  \tabularnewline
    \end{tabular}%
\end{table}

\begin{table*}
\caption{Best settings for each model derived by checking performance on validation set}
\label{best_settings}
    \centering
    \begin{tabular}{l|r||r||r|r}
    \multicolumn{1}{l}{}&\multicolumn{1}{|c}{\begin{tabular}[x]{@{}c@{}}\textbf{}\\\end{tabular}}&\multicolumn{1}{|c}{\begin{tabular}[x]{@{}c@{}}\textbf{}\\\end{tabular}}&\multicolumn{2}{|c}{\begin{tabular}[x]{@{}c@{}}\textbf{}\\\end{tabular}}\tabularnewline
    \multicolumn{1}{l}{\textbf{RNN Models}}&\multicolumn{1}{|c}{\begin{tabular}[x]{@{}c@{}}\textbf{Hidden size}\\\end{tabular}}&\multicolumn{1}{|c}{\begin{tabular}[x]{@{}c@{}}\textbf{layers}\\\end{tabular}}&\multicolumn{1}{|c}{\begin{tabular}[x]{@{}c@{}}\textbf{optimizer}\\\end{tabular}}&\multicolumn{1}{|c}{\begin{tabular}[x]{@{}c@{}}\textbf{learning rate}\\\end{tabular}}\tabularnewline
\hline
        \textit{LSTM} & $1024$ & $2$ & $Adam$ & $1e-4$  \tabularnewline
        \textit{Stack-RNN} & $512$ & $1$ & $SGD$ & $1e-3$  \tabularnewline
        \textit{NTM} & $256$ & $1$ & $SGD$ & $2e-3$   \tabularnewline
        \textit{Transformer} & $1024$ & $3$ & $Adam$ & $2e-5$
        \tabularnewline
        \textit{NSTM (\textbf{ours})} & $8$ & $1$ & $SGD$ & $1e-2$  \tabularnewline
    \end{tabular}
\end{table*}

\begin{table*}
\caption{Percentage of correctly classified strings of RNNs trained on the $D_4$ language (across $3$ trials). We report the mean accuracy for each model on both the training and validation sets. We also report the average number of epochs required by each model to get the best score (as reported in previous tables) on the training and validation sets.}
\label{average_epochs}
    \centering
    \begin{tabular}{l|r||r||r|r}
    \multicolumn{1}{l}{}&\multicolumn{1}{|c}{\begin{tabular}[x]{@{}c@{}}\textbf{$Train$}\\\end{tabular}}&\multicolumn{1}{|c}{\begin{tabular}[x]{@{}c@{}}\textbf{$Validation$}\\\end{tabular}}&\multicolumn{2}{|c}{\begin{tabular}[x]{@{}c@{}}\textbf{$Number of Epochs$}\\\end{tabular}}\tabularnewline
    \multicolumn{1}{l}{\textbf{RNN Models}}&\multicolumn{1}{|c}{\begin{tabular}[x]{@{}c@{}}\textbf{Mean}\\\end{tabular}}&\multicolumn{1}{|c}{\begin{tabular}[x]{@{}c@{}}\textbf{Mean}\\\end{tabular}}&\multicolumn{1}{|c}{\begin{tabular}[x]{@{}c@{}}\textbf{Train}\\\end{tabular}}&\multicolumn{1}{|c}{\begin{tabular}[x]{@{}c@{}}\textbf{Validation}\\\end{tabular}}\tabularnewline
\hline
        \textit{LSTM} & $119$ & $151$ & $90$ & $140$  \tabularnewline
        \textit{Stack-RNN} & $100$ & $99.99$ & $100$ & $130$  \tabularnewline
        \textit{NTM} & $100$ & $99.99$ & $89$ & $115$   \tabularnewline
        \textit{Transformer} & $100$ & $98.00$ & $135$ & $195$
        \tabularnewline
        \textit{NSTM (\textbf{ours})} & $100$ & $99.99$ & $100$ & $108$  \tabularnewline
    \end{tabular}
\end{table*}

\begin{table}
\caption{Percentage of correctly classified strings of the NSTM trained on the $D_4$ language (across $3$ trials). We report the change in model performance as a function of the number of hidden units ($H$). We report the mean accuracy on long strings ($n=1000$).}
\label{D4performance_neurons_nstm}
    \centering
    \begin{tabular}{l|r||r||r|r}
    \multicolumn{1}{l}{}&\multicolumn{1}{|c}{\begin{tabular}[x]{@{}c@{}}\textbf{}\\\end{tabular}}&\multicolumn{1}{|c}{\begin{tabular}[x]{@{}c@{}}\textbf{}\\\end{tabular}}&\multicolumn{2}{|c}{\begin{tabular}[x]{@{}c@{}}\textbf{}\\\end{tabular}}\tabularnewline
    \multicolumn{1}{l}{\textbf{RNN Model}}&\multicolumn{1}{|c}{\begin{tabular}[x]{@{}c@{}}\textbf{H=4}\\\end{tabular}}&\multicolumn{1}{|c}{\begin{tabular}[x]{@{}c@{}}\textbf{H=6}\\\end{tabular}}&\multicolumn{1}{|c}{\begin{tabular}[x]{@{}c@{}}\textbf{H=8}\\\end{tabular}}&\multicolumn{1}{|c}{\begin{tabular}[x]{@{}c@{}}\textbf{H=12}\\\end{tabular}}\tabularnewline
\hline
        \textit{NSTM (\textbf{ours})} & $95.00$ & $97.00$ & $99.40$ & $98.99$  \tabularnewline
    \end{tabular}
\end{table}

\section{Appendix B: Detailed Discussion and Conclusion}
\label{sec:detailed_conclusion}

In this work, we have introduced a new class of formal models, the Neural State Turing Machine (NSTM), which are interpretable and contains higher-order/tensor weights that can model any Turing Machine. To construct a Turing-complete NSTM, we incorporated a transition table and tape symbols directly into the model's $3^{rd}$ order weights. This construction can be done in the following way: 
\textbf{1)}
transition rules and tape symbols are directly encoded into the unbounded precision of state neurons, 
\textbf{2)} for bounded precision state neurons, an additional state (the halting state) is integrated to capture illegal settings.
\textbf{3)} using $n^{th}$ order feedforward connections in a bounded weight and precision NSTM, which is also Turing Complete.
The main contribution of this work is to analyze the evolution of an NSTM using tensor products between state and action neurons with and without recurrence. We proved the Turing completeness of a $6$-neuron unbounded precision recurrent neural network (RNN), which is the smallest Turing complete RNN of any order to date; prior results \cite{chung2021turing} show that for a $6$ state, $4$ symbol UTM, there exists a $40$-neuron unbounded precision RNN and $54$-neuron bounded precision RNN augmented with two stack like memory. 

We further analyzed the relationship between the number of neurons, the precision of NSTM, and the action neurons that control the tape when simulating a TM and show equivalence using the tensor product of ranks 1 and 2. Furthermore, we show that the differentiable version of the NSTM with $13$ bounded neurons obeys the associative and cumulative law and can simulate any TM in real-time $O(T)$. We later proved that a feedforward version of the NSTM can simulate any TM with bounded precision and weights using an nth-order connection; we use only $6$ neurons to show Turing completeness. Prior results with feedforward models, even with transformers, assume infinite weights and arbitrary precision. We believe this to be a substantial result. Furthermore, on Dyck languages, we experimentally demonstrated that the NSTM can be stably trained using real-time recurrent learning (RTRL) and can even outperform the Neural Turing Machine (NTM) \cite{graves2014neural} and transformers \cite{vaswani2017attention} when tested on longer strings. 

This work provides a theoretical bound for RNNs with finite precision and weights. We also provide a partial way to train these networks using a forward propagation approach such as RTRL \cite{williams1989rtrl}. There are several avenues for improvement, and we discuss these and their limitations in the Appendix. One computational benefit of the NSTM is that it requires fewer neurons and training samples to recognize complex grammars. However, optimizing tensor weights with backpropagation is challenging. Various techniques are needed to stabilize learning, as shown in a few prior efforts, \cite{mali2020neural}, such as stabilizing gradients, using discrete activations with different learning rules \cite{zeng1994discrete, williams1995gradient, yin2018understanding}, or using forward gradients \cite{kag2021training}. Also, training with RTRL is time-consuming and not scalable when applied to larger networks. Therefore, it is worth investigating computationally efficient approaches that serve as approximations of RTRL \cite{tallec2017uoro, menick2020practical}. However, our results in this work are mainly theoretical, but if one could train NSTMs efficiently, the potential negative societal impact would be biasedness. Since the NSTM operates based on rules obtained from computational models, the rules generated by the NSTM are heavily biased toward the data on which it is trained. Hence, it will be essential to design an ethical framework to avoid such data-specific bias in the future.

\section{Appendix C: On Limitations and Societal Impact}

We construct an empirical model that shows the advantage of the NSTM compared to other modern-day ANNs applied to problems in natural language processing. This work theoretically shows that the smallest neural network with bounded precision is Turing complete. However, there are several key avenues for improvement for NSTM. First, the current model is challenging to optimize using standard learning algorithms such as backpropagation (of errors) through time (BPTT) or truncated BPTT. Thus, it is crucial to design stable learning approaches that are computationally efficient and work effectively with higher-order connections. In this work, we trained our NSTM using real-time recurrent learning (RTRL), which works better than BPTT but is, unfortunately, computationally expensive and thus limits the scalability of NSTM. Second, modern regularization approaches such as recurrent dropout and layer norm do not work well with higher-order/tensor connections. Therefore, stabilizing gradients and avoiding overfitting becomes much more challenging with (memory-augmented) models such as the NSTM. Finally, modern/popular optimization update rules such as RMSprop or Adam do not function well with the NSTM; hence, all our models are trained with stochastic gradient descent. Given that this current study focused on deriving the true computational power of an RNN with bounded precision and weights, our future work will examine how to construct a stable empirical model compatible with mainstream optimization tools (or variations of them) that still works with complex algorithm patterns. 

Our work is primarily theoretical, but note that we have proposed a model that, if trained stably, marks an important step towards constructing interpretable neuro-symbolic intelligent systems from which explainable rules may be extracted. Such systems can be eventually used in high-risk domains such as clinical trials and medicine, functioning for the betterment of society. However, ensuring that an ethical framework is designed for properly evaluating the fairness and biasedness of the rules generated from these complex memory-augmented systems is essential.  

\end{document}